\documentclass[preprint,12pt]{elsarticle}
\usepackage[utf8]{inputenc}
\usepackage{amsthm}
\usepackage{amsmath}
\usepackage{amsfonts}
\usepackage{amssymb}
\usepackage{graphicx}
\usepackage{epstopdf}
\usepackage{color}
\usepackage{multirow}
\usepackage{mathtools}
\usepackage{xcolor}
\usepackage{a4wide}
\usepackage{bm}
\usepackage{caption}
\usepackage{subcaption}
\usepackage[
  separate-uncertainty = true,
  multi-part-units = repeat
]{siunitx}
\usepackage{rotating}
\usepackage{verbatim}
\usepackage{lineno,hyperref}
\modulolinenumbers[5]
\usepackage{tikz}
\tikzstyle{noeud} = [circle, draw, fill=white, inner sep=2pt]
\journal{Theoretical Computer Science}
\newtheorem{theorem}{Theorem}[section]
\newtheorem{lemma}{Lemma}[section]

\newdefinition{remark}{Remark}[section]
\newtheorem{observation}{Observation}[section]
\newdefinition{definition}{Definition} \newdefinition{example}{Example}

\usepackage[noend,linesnumbered,ruled,lined]{algorithm2e}
\SetAlFnt{\scriptsize}
\usepackage{amsfonts}
\usepackage{makecell}

\usepackage[usenames,dvipsnames]{pstricks}
\usepackage{pstricks-add}
\usepackage{epsfig}
\usepackage{pst-grad} 
\usepackage{pst-plot} 
\usepackage[space]{grffile} 
\usepackage{etoolbox} 
\makeatletter 
\patchcmd\Gread@eps{\@inputcheck#1 }{\@inputcheck"#1"\relax}{}{}
\makeatother
\usepackage{mathtools}

\usepackage{graphicx}
\usepackage{tabularx}
\usepackage{textcomp}
\usepackage{xcolor}
\usepackage[labelformat=simple]{subcaption}

\usepackage{booktabs}

\usepackage{pgf,tikz,pgfplots}
\pgfplotsset{compat=1.15}
\usepackage{mathrsfs}
\usepackage{hyperref}
\usepackage{multicol}
\SetAlFnt{\scriptsize}
\usepackage{amsfonts}
\usepackage{makecell}
\usepackage{bm}
\usepackage{multirow}

\usepackage{multirow}
\usepackage{array}
\usepackage{makecell}
\usepackage{graphicx}
\usepackage{orcidlink}
\newcolumntype{P}[1]{>{\raggedright\arraybackslash}p{#1}}
\newtheorem{problem}{Problem}

\SetCommentSty{mycommfont}

\journal{Arxiv}

\usepackage{a4wide}
\begin{document}
\usetikzlibrary{arrows}
\begin{frontmatter}


\title{Black Hole Search: Dynamics, Distribution, and Emergence\footnote{A preliminary version of this work appeared in the proceedings of 27th Edition of SSS 2025 \cite{BHS_SSS25}.}}

\author[inst1]{Tanvir Kaur\orcidlink{0009-0002-3651-994X}}

\author[inst1]{Ashish Saxena\orcidlink{0009-0007-4767-8862}}

\author[inst2]{Partha Sarathi Mandal\orcidlink{0000-0002-8632-5767}}

\author[inst1]{Kaushik Mondal\orcidlink{0000-0002-9606-9293}}
\cortext[mycorrespondingauthor]{Corresponding author}
\ead{kaushik.mondal@iitrpr.ac.in}

\affiliation[inst1]{organization={Department of Mathematics},
            addressline={Indian Institute of Technology Ropar}, 
            city={Rupnagar},
            postcode={140001}, 
            state={Punjab},
            country={India}}
            
\affiliation[inst2]{organization={Department of Mathematics},
            addressline={Indian Institute of Technology Guwahati}, 
            city={Guwahati},
            postcode={781039}, 
            state={Assam},
            country={India}}


 \begin{abstract}
  A black hole is a malicious node in a graph that destroys resources entering into it without leaving any trace. The problem of Black Hole Search (BHS) using mobile agents requires that at least one agent survives and terminates after locating the black hole. Recently, this problem has been studied on 1-bounded 1-interval connected dynamic graphs \cite{BHS_gen}, where there is a footprint graph, and at most one edge can disappear from the footprint in a round, provided that the graph remains connected. In this setting, the authors in \cite{BHS_gen} proposed an algorithm that solves the BHS problem when all agents start from a single node (rooted initial configuration). They also proved that at least $2\delta_{BH} + 1$ agents are necessary to solve the problem when agents are initially placed arbitrarily across the nodes of the graph (scattered initial configuration), where $\delta_{BH}$ denotes the degree of the black hole. In this work, we present an algorithm that solves the BHS problem using $2\delta_{BH} + 17$ initially scattered agents. Our result matches asymptotically with the rooted algorithm of \cite{BHS_gen} under the same model assumptions.
  \vspace{0.15cm}
  
  Further, we study the Eventual Black Hole Search (\textsc{Ebhs}) problem, in which the black hole may appear at any node and at any time during the execution of the algorithm, destroying all agents located on that node at the time of its appearance. However, the black hole cannot emerge at the home base in round~0, where the home base is the node at which all agents are initially co-located. Once the black hole appears, it remains active at that node for the rest of the execution. This problem has been studied on static rings~\cite{Bonnet25}; here we extend it to arbitrary static graphs and provide a solution using four  agents. Moreover, it does not require any knowledge of global parameters or additional model assumptions.
\end{abstract}


\end{frontmatter}

\section{Introduction}
Mobile agents are software entities that work in networking environments. In real-life scenarios, these networking environments are vulnerable to various risks. A lot of research is concentrated on safeguarding network sites (hosts) against malicious agents and, conversely, defending agents from attacks by the hosts. In this work, we address a particularly dangerous type of host, referred to as a \textit{black hole}, which is a network site that eliminates visiting agents immediately, leaving no trace. The presence of a black hole in systems that support code mobility is not uncommon. For example, a site can unintentionally become a black hole due to an unnoticed failure or if a virus, installed unknowingly, discards any incoming communication (such as marking it as spam). In these instances, identifying this harmful host is crucial. This problem of identifying such a harmful host in the graph through a team of mobile agents is known as the Black Hole Search (BHS) problem. Precisely, an agent is said to have located the black hole (thereby solving the BHS problem) if it terminates its algorithm at a safe node\footnote{A safe node is a node that is not the black hole.} $v$ that is adjacent to the black hole and outputs the port number that leads from $v$ to the black hole \cite{dyn_ring_jrnl, bhattacharya_2023, Adri_tori, BHS_gen}.

The BHS problem has been extensively studied under various settings, depending on the agents’ communication capabilities, synchronization assumptions, and knowledge of the graph. Most of the literature focuses on static graphs~\cite{Pelc_2005, Paola_2010, Shantanu_2011, Paola_2004}, while more recent work considers dynamic graphs such as dynamic rings~\cite{dyn_ring_jrnl}, dynamic cactuses~\cite{bhattacharya_2023}, dynamic tori~\cite{Adri_tori}, and dynamic graphs with arbitrary underlying topology~\cite{BHS_gen}. In~\cite{BHS_gen}, the authors solve BHS for a rooted initial configuration, whereas in this work, we consider the more general scattered initial configuration.

Another variant of BHS allows the black hole to appear at an arbitrary time during the execution rather than being present from the beginning~\cite{Bonnet25}. This \emph{eventually emerging black hole} can destroy all agents located on its node at the moment of its appearance, making the problem significantly more challenging. In this work, we also study eventually emerging black hole search (\textsc{Ebhs}) problem on arbitrary static graphs.

\subsection{The Model and the Problem}

\noindent \textbf{Dynamic graph model:} 
A dynamic network is modeled as a \emph{time-varying graph (TVG)}, denoted by \( \mathcal{G} = (V, E, T, \rho) \), where \( V \) is a set of nodes, \( E \) is a set of edges, \( T \) is the temporal domain, and \( \rho : E \times T \rightarrow \{0, 1\} \) is the \emph{presence function}, which indicates whether a given edge is present at a given time. Here, $\rho(e,t)=0$ indicates that the edge $e\in E$ is not present at round $t\in T$, and $\rho(e,t)=1$ indicates that the edge $e\in E$ is present at round $t\in T$. The static graph \( G = (V, E) \) is referred to as the \emph{underlying graph} (or \emph{footprint}) of the TVG \( \mathcal{G} \), where \( |V| = n \) and \( |E| = m \). For a node \( v \in V \), let \( E(v) \subseteq E \) denote the set of edges incident on \( v \) in the footprint. The \emph{degree} of node \( v \) is defined as \( \delta_v = |E(v)| \), and the \emph{maximum degree} of \( G \) is given by \( \Delta = \max_{v \in V} \delta_v \). The nodes in \( V \) are assumed to be \emph{anonymous}, and have no storage unless it is mentioned. Each edge incident to a node \( v \) is locally labeled with a \emph{port number}. This is defined by a bijective function \( \lambda_v : E(v) \rightarrow \{0, \ldots, \delta_v - 1\} \), which assigns a distinct label to each incident edge of \( v \). No further assumptions are made about the labeling. Under the assumption that time is discrete, the TVG \( \mathcal{G} \) can be viewed as a sequence of static graphs \( \mathcal{S}_{\mathcal{G}} = \mathcal{G}_0, \mathcal{G}_1, \ldots, \mathcal{G}_r, \ldots \), where each \( \mathcal{G}_r = (V, E_r) \) denotes the \emph{snapshot} of \( \mathcal{G} \) at round \( r \), with \( E_r = \{ e \in E \mid \rho(e, r) = 1 \} \). The set of edges not present at round \( r \) is denoted by \( \overline{E}_r = E \setminus E_r \subseteq E \). If $E_r=E$ for every round $r$, then $\mathcal{G}$ is called \underline{\emph{static graph}}.

Based on this representation, we recall the following definitions, commonly used in prior works~\cite{Kuhn_2010, GOTOH20211, Adri_tori, BHS_gen,bhattacharya_2023}.

\begin{definition}[\cite{Kuhn_2010}]
\textbf{(1-Interval Connectivity)} A dynamic graph \( \mathcal{G} \) is said to be \emph{1-interval connected} (or \emph{always connected}) if, for every snapshot \( G_t \in \mathcal{S}_{\mathcal{G}} \), the graph \( G_t \) is connected.
\end{definition}

\begin{definition}[\cite{GOTOH20211}]
    \textbf{(\( 1 \)-Bounded 1-Interval Connectivity)} A dynamic graph \( \mathcal{G} \) is said to be \emph{\( 1 \)-bounded 1-interval connected} if it is always connected and, for all \( t \in T \), the number of missing edges satisfies \( |\overline{E}_t| \leq 1 \).
\end{definition}

In particular, suppose an edge $e$ disappears from $G$ at some round, resulting in the graph $G \setminus \{e\}$, which remains connected as per the definition of 1-Interval Connectivity. By the time another edge $e'$ disappears, the edge $e$ must reappear as it is 1-bounded. In this work, we consider 1-bounded 1-interval connected TVGs. 

\vspace{0.1cm}
\noindent\textbf{Agent model}: Initially, $l$ many agents are scattered across the safe nodes of the graph. Note that, in a scattered configuration, one node may contain multiple agents as well. Each agent has a unique identifier assigned from the range $[1, n^c]$, where $c > 1$ is a constant. Each agent knows its ID. Agents do not have any prior knowledge of $l$, or any other graph parameters like $n$, $\delta_{BH}$, $\Delta$. Here, $\delta_{BH}$ denotes the degree of the black hole. The agents are equipped with $O(\log n)$ bits of memory and have access to the whiteboard. When an agent visits a node, it knows the degree of that node in the footprint $G$. However, when an agent is at a node $v$, it cannot detect if any edge incident to $v$ is missing. To be precise, say an agent $a$ reaches a node $v$ at round $t$ and $e_v$ be an edge associated with $v$, then the agent $a$ can not determine $\rho(e_v,\,t)$. It can only understand this by traversing through the edge. If it tries to move through an edge and is unable to do so, it understands that the edge is missing. In other words, it understands that $\rho(e_v,\,t)=0$ at the beginning of the next round $t+1$. Otherwise ($\rho(e_v,\,t)=1$), it can successfully move through that edge. An agent learns the port number through which it enters into a node. If two agents move via an edge, they can not see each other during the movement.
\\
Our algorithm runs in synchronous rounds. In each round, an agent performs one \textit{Look-Compute-Move} (LCM) cycle at a safe node:
\begin{itemize}
    \item \textit{Look}: The agent reads the contents of the whiteboard of the node it occupies and sees if there are other agents at the same node. Each agent can read all the variable values of all other agents present at the same node. The agent also understands whether it had a successful or an unsuccessful move in the last round.
    \item \textit{Compute}: On the basis of the information obtained in the Look phase, the agent decides whether to move through the port or not in this round. The agent may write some information on the whiteboard as per the computation.
    \item \textit{Move}: In compute, if an agent computes some port $p$ to move, it tries to move through the port $p$; if the corresponding edge is present in that round, the agent reaches the adjacent node; otherwise, it remains at the current node.
\end{itemize}

We first define a black hole and an eventually emerging black hole following prior work, and then present the formal problem statement considered in this paper.

\begin{definition}[Black Hole \cite{Dobrev2007BHSRing}]
\emph{A \emph{black hole} is a node that destroys every agent that enters it without leaving any trace and is present from the beginning of the execution (i.e., at round $0$).}
\end{definition}

\begin{definition}[Eventually Emerging Black Hole \cite{Bonnet25}]
\emph{An \emph{eventually emerging black hole} is a node, say $v_{BH}$, that may not be a black hole initially but becomes one at some round $r \ge 0$ and remains a black hole for the rest of the execution. Any agent located at $v_{BH}$ at the end of round $r-1$ that does not leave before round $r$, as well as any agent that enters $v_{BH}$ at a round $r' \ge r$, is destroyed without leaving any trace.}
\end{definition}
\noindent If the black hole emerges at round $r=0$, the definition reduces to the black hole. In this work, we study the following two problems.
\begin{problem}
    [\textbf{1-Black Hole Search in Dynamic Graphs}]\label{prb:1} \emph{Let \( \mathcal{G} \) be a \( 1 \)-bounded 1-interval connected TVG. There is a node $v_{BH}$ in $G$ that is a black hole. Initially, $l$ many agents are located at an arbitrary subset of the safe nodes in $G$. The agents have no prior knowledge of the graph or its global parameters. The objective is for at least one agent to survive, locate the black hole, and terminate.}
\end{problem}

\begin{problem}[\textbf{Eventual Black Hole Search in Static Graphs}]\label{prb:2}
 \emph{Let \( \mathcal{G} \) be a static graph. There is a node in $G$ which is an eventually emerging black hole $v_{BH}$. Initially, $l$ many agents are co-located at a designated home base node $\mathcal{H}$. If node $\mathcal{H}$ is an emerging black hole, it cannot emerge at round 0, but it may emerge at any round $r\geq 1$. The goal is for at least one agent to survive, locate $v_{BH}$, and ultimately terminate.}
\end{problem}

In this work, we denote Problem \ref{prb:1} by 1-BHS, and Problem \ref{prb:2} by \textsc{Ebhs}. We define the notion of time complexity for both problems as follows.

\medskip
\noindent\textbf{Time Complexity:}
For \emph{1-BHS}, the time complexity is measured from round $0$ until at least one agent detects the location of the black hole.  
For \textsc{Ebhs}, the time complexity is measured from the round $r$ at which the black hole emerges until at least one surviving agent identifies its location.

\subsection{Related Work}
The BHS problem was originally introduced in the context of static networks \cite{Paola_2006}, where the goal is for a team of mobile agents to locate the black hole that destroys any agent entering it, without any observable trace. The problem is thoroughly explored in static graphs under various assumptions regarding agent capabilities, memory limitations, communication models, and initial placements \cite{Pelc_2005, Paola_2010, Shantanu_2011, Paola_2006, MarkouP12, complexity_black_hole, bhs_tokens, Paola_2002, Paola_2004}.

In recent years, research has increasingly turned to dynamic graphs, where the network topology changes over time. Within this setting, most existing work studies both the problems of 1-BHS and $f$-BHS \footnote{Similar to the 1-Bounded 1-Interval Connectivity, \( f \)-bounded 1-Interval connected dynamic graph \( \mathcal{G} \) is the graph that is always connected and, for all \( t \in T \), the number of missing edges satisfies \( |\overline{E}_t| \leq f \). Note that, for $f$-bounded 1-interval connected graphs, $\rho(e,t)$ can take value 0 for at most $f$ many edges $e\in E$ at a round $t\in T$.  Analogously, if \( \mathcal{G} \) is an \( f \)-bounded 1-interval connected dynamic graph with a black hole located at some node in its footprint \( G \), then the black hole search problem is referred to as $f$-BHS.}. These problems are explored across a range of graph classes. The authors in \cite{dyn_ring_jrnl} first investigated the problem of 1-BHS on dynamic rings. The authors provided a characterization of the problem based on two cases: face-2-face communication model and when the nodes of the graph have storage in the form of a whiteboard. When the three agents are initially co-located, they can solve the 1-BHS problem in the whiteboard and pebble communication model in $\Theta(n^{1.5})$ rounds even if the agents are anonymous. If agents can communicate only when they are on the same node, the authors show that the 1-BHS problem can be solved in $\Theta(n^2)$ moves and rounds if the agents have unique IDs and is impossible to solve if the agents are anonymous. Further, the authors studied the problem with scattered agents in rings. They also show that it is impossible to solve the 1-BHS problem with scattered agents only using a face-to-face communication model.

Further, in \cite{bhattacharya_2023}, the authors investigate both the problems of 1-BHS and $f$-BHS on the cactus graphs. For 1-BHS, the authors provide a lower bound of $\Omega(n^{1.5})$ rounds, and an upper bound of $O(n^2)$ rounds with the help of three agents. Further, for $f$-BHS, the authors provided the impossibility of finding the black hole with $f+1$ agents. Later in \cite{Adri_tori}, the authors consider the problem of $(n+m)$-BHS problem where the tori is of size $n\times m$. They consider both rooted and scattered initial configurations by the agents. They provided an algorithm that solves $(n+m)$-BHS using $n+3$ agents in $O(n\cdot m^{1.5})$ rounds, where the size of tori $n\times m$ ($3\leq n \leq m$). Further, using $n+4$ agents, they provided an algorithm that solves $(n+m)$-BHS in $O(m\cdot n)$ rounds. With scattered agents, they provide two algorithms. The first one uses $n+6$ agents to solve the $(n+m)$-BHS problem in $O(n\cdot m^{1.5})$ rounds. The second one uses $n+7$ agents to solve the $(n+m)$-BHS problem in $O(n\cdot m)$ rounds. Recently, the problems of 1-BHS and $f$-BHS are studied for general graphs in \cite{BHS_gen}. The authors prove that the 1-BHS problem cannot be solved with fewer than $2\delta_{BH}$ agents arbitrarily placed at safe nodes, even when both agents and nodes are equipped with $O(\log n)$ bits of memory. They also show that, for the $f$-BHS problem, no solution exists with $2f+1$ co-located agents, even when memory and storage are unbounded. While the paper presents algorithms for both models, these solutions assume a \textit{rooted initial configuration}, where all agents begin at a common safe node. In this work, we consider a more general scenario in which agents are arbitrarily placed across safe nodes in the graph. We adopt the same dynamic graph model as in \cite{BHS_gen}.

\textsc{Ebhs} is introduced on static rings \cite{Bonnet25}. The authors present a time-optimal solution using four agents. They further show that time optimality can be achieved with three agents, provided either $O(1)$ node storage is available, or the agents know $n$. In addition, they establish an impossibility result for two agents.

There are also studies on Byzantine black holes, in which the black hole may arbitrarily choose whether to destroy visiting agents in each round. In \cite{GoswamiBD024}, the authors solve ring exploration in the presence of a Byzantine black hole for both rooted and scattered initial configurations. Subsequently, \cite{Adri_ByzBH} extends the study to arbitrary graphs with agents starting from a rooted configuration. This setting is strictly more general and harder than \textsc{Ebhs}, and consequently requires more agents.

In contrast, in this work, we consider eventually emerging black holes on arbitrary static graphs.

\subsection{Our Contribution}
In this work, we provide an algorithm \textsc{Algo\_Bhs\_Scattered} that solves the problem of 1-BHS using $2\delta_{BH}+17$ agents starting from a scattered initial configuration, where $\delta_{BH}$ is the degree of the black hole in the footprint. The time required by the agents to solve this case is $O(m^2)$, and each agent requires $O(\log n)$ bits of memory to run this algorithm, where $n$ denotes the number of nodes in the graph and $m$ denotes the number of edges in the graph (refer Theorem \ref{th:scattered}). We match the asymptotic time complexity of the solution of 1-BHS from a rooted initial configuration in \cite{BHS_gen} under the same model assumptions, using the same amount of memory per agent and storage per node. 
Table \ref{fig:literature} illustrates comparison of our results with the results of \cite{BHS_gen}.

\begin{table}[ht!]
\vspace{-0.35cm}
\centering
\scriptsize
\renewcommand{\arraystretch}{1.25}
\setlength{\tabcolsep}{6pt}
\begin{tabular}{|c|c|c|c|c|c|c|}
\hline
 & \textbf{IC} & \textbf{Problem} & $\bm{l}$ & \textbf{Node storage} & \textbf{Agent memory} & \textbf{Time complexity} \\ 
\hline\hline

\cite{BHS_gen} & Rooted & 1-BHS & 9 
& $O(\log n)$ & $O(\log n)$ & $O(m^2)$ \\
\hline
\cite{BHS_gen} & Scattered & $1$-BHS & $2\delta_{BH}$ 
& $O(\log n)$ & $O(\log n)$ 
& Impossible \\
\hline
\makecell{This\\work} & Scattered & 1-BHS & $2\delta_{BH}+17$ 
& $O(\log n)$ & $O(\log n)$ & $O(m^2)$ \\
\hline

\end{tabular}
\vspace{0.15cm}
\caption{Summary of existing results on 1-BHS on general graphs and our result. IC denotes the initial configuration and $l$ denotes the number of agents.}
\label{fig:literature}
\end{table}

\begin{table}[h]
\centering
\scriptsize
\renewcommand{\arraystretch}{1.25}
\setlength{\tabcolsep}{6pt}
\begin{tabular}{|c|c|c|c|c|c|c|}
\hline
Static Graph &
IC &
$l$ &
Pebble &
\begin{tabular}[c]{@{}c@{}}Node\\Storage\end{tabular} &
Knowledge &
Complexity \\ \hline\hline

\multirow{5}{*}{Ring \cite{Bonnet25,bonnet_2026}} &
\multirow{5}{*}{Rooted} &
2 & Yes & Infinite & $n$ & Impossible \\ \cline{3-7}

& & 4 & No  & No     & No  & $\Theta(n)$ \\ \cline{3-7}

& & 3 & Yes & $O(1)$ & No  & $\Theta(n)$ \\ \cline{3-7}

& & 3 & No  & No     & $n$ & $\Theta(n)$ \\ \cline{3-7}

& & 3 & No  & No     & No  & $O(n^2)$ \\ \hline

\multirow{2}{*}{\begin{tabular}[c]{@{}c@{}}General Graph\\(This Work)\end{tabular}} &
\multirow{2}{*}{Rooted} &
4 & No & No & No & $\mathrm{poly}(n)^*$ \\ \cline{3-7}

& & 4 & No & No & $n$ & length of the UXS \\ \hline

\end{tabular}
\caption{Summary of existing results of \textsc{Ebhs} and our results on \textsc{Ebhs} and our results. IC denotes the initial configuration and $l$ denotes the number of agents. $^*$If the black hole emerges within $\mathrm{poly}(n)$ time.}
\label{fig:literatureebhs}
\end{table}

We provide an algorithm to solve \textsc{Ebhs} using 4 agents on general graphs in polynomial time, considering that the black hole emerges in the graph within polynomial time (refer Theorem\ref{thm:ebhs} and Remark \ref{re:1}). With the knowledge of $n$, the runtime becomes the length of the UXS on any graph of size $n$ (refer Remark \ref{re:2}). A summary of existing results and our results are presented in Table \ref{fig:literatureebhs}. Further, we discuss how any single agent exploration algorithm on arbitrary static graph can be simulated using our 4 agents strategy to solve the problem and this shows a tradeoff between time complexity and memory/storage requirement (refer Remark \ref{re:4}).


\subsection{Preliminaries}\label{sec:preliminaries}
In this section, we recall the strategy which is used in the existing literature \cite{BHS_gen}, known as the cautious movement of agents. For the sake of completeness, we provide the strategy of cautious movement by three agents.\\

\noindent\textbf{Cautious walk:} Cautious walk is a movement strategy that ensures that one agent stays alive if a group of agents performs this movement while exploring an edge $(u,v)$ from a node $u$ without knowing that $v$ is the black hole. Let $a_1$ (leader), $a_2$ (first helper), and $a_3$ (second helper) be the three agents at a safe node $u$ and these agents want to traverse through the edge $(u,v)$ to reach $v$. Since the agents are not aware whether $v$ is a safe node or a black hole, it is necessary to first verify this. Thus, the first helper $a_2$ attempts to move through $(u,v)$. If $a_2$ is successful in moving through $(u,v)$ at some round $t$, then in the round $t+1$, $a_3$ attempts to move through $(u,v)$ while $a_2$ (if alive) attempts to move through $(v,u)$. If $a_3$ is successful in its movement and $a_2$ is not present at the current node, then $a_1$ concludes that $v$ is a black hole. If $a_3$ is successful in its movement and $a_2$ is also present at $u$, then it is concluded that $v$ is a safe node, and so, $a_1$ and $a_2$ move through the edge $(u,v)$. In this way, cautious movement can be executed by three agents.\\

\noindent\textbf{Idea of the algorithm for rooted agents \cite{BHS_gen}:} The authors first provide an exploration strategy using $3$ agents that ensures the exploration of a TVG when there is no black hole. Further, they replace each edge movement of the exploration strategy with the cautious walk. The role of one agent is taken over by three agents that are sufficient to do the cautious movement. Thus, a total of 9 agents can solve 1-BHS on an arbitrary graph from a rooted configuration. The final theorem is as follows:
\begin{theorem}
    \cite{BHS_gen} The problem of 1-BHS can be solved by 9 agents starting from a rooted initial configuration in $O(m^2)$ rounds, requiring $O(\log \delta_v)$ storage per node.
\end{theorem}

We now discuss the challenges that arise when attempting to solve the 1-BHS problem from a scattered initial configuration using a lesser number of agents.\\

\noindent\textbf{Challenges with fewer agents:}
An intuitive approach to solve the 1-BHS problem is to first design an exploration strategy that guarantees every node of the graph is visited by at least one agent, even under adversarial edge deletions (when the graph is 1-bounded 1-interval connected). However, since the graph contains a black hole, it becomes essential that every movement performed by an agent is cautious. If agents move without caution, they may all fall into the black hole and die. Therefore, solving the 1-BHS problem requires a combination of two strategies: cautious movement and exploration of dynamic graphs.

We assume the graph to be 1-bounded 1-interval connected. As shown in \cite{BHS_gen}, it is impossible to solve the 1-BHS problem on such graphs using only $2\delta_{BH}$ agents. Let us assume that we have $2\delta_{BH}+1$ agents. Suppose  $2\delta_{BH}$ agents are initially close by the black hole and all die soon after, keeping information on the adjacent nodes. However, since only one agent remains alive in the graph, the problem is solved only if this agent reaches one of those $\delta_{BH}$ nodes and reads the information. Unfortunately, the adversary can always prevent this agent’s movement by making an edge disappear, and hence it might be impossible to guarantee exploration. Now, let us assume that we have $2\delta_{BH}+3$ agents, and so there are 3 alive agents even after $2\delta_{BH}$ agents die. Moreover, let us assume that those 3 alive agents are together. We know 3 rooted agents can explore 1-bounded 1-interval graph if there is no black hole as provided in \cite{BHS_gen}. 
Note that the death of $2\delta_{BH}$ agents ultimately isolates the black hole, if the information of all the ports corresponding to the edges leading to the black hole are kept at those nodes.
The problem is that those 3 agents do not know when to start the algorithm of exploration by 3 agents, as they do not know if and when all the remaining $2\delta_{BH}$ agents die. Further, as we start from a scattered initial configuration, those 3 agents may not be together, which adds to the difficulty. Note that none of the above discussions are formal; we are unable to provide a  tighter bound than that of $2\delta_{BH}$ with scattered agents as provided in \cite{BHS_gen} as of now.

We start with a number of extra agents such that the adversary is either bound to form a group of 9 agents, where any such group can deterministically transition to the existing algorithm of rooted 1-BHS without bothering about the situation of the other agents. If such a group is not formed, then at least $2\delta_{BH}+1$ agents explore the graph, where at most $2\delta_{BH}$ may die, and the other agent finds the black hole. 

\section{Algorithm to solve 1-BHS}\label{sec:1BHS}
In this section, we provide an algorithm that uses $2\delta_{BH}+17$ agents to solve 1-BHS when the agents are initially present arbitrarily at the safe nodes of the graph, and each node has $O(\log n)$ storage. We begin with a high-level idea of the algorithm.

\subsection{High-level Idea of our Algorithm}
Let us first consider that there is sufficient storage at each node to store the information corresponding to all the $2\delta_{BH}+17$ agents. Each agent, in this case, simply runs its own DFS unless it is stuck by a missing edge or dies in the black hole. If agents do not move cautiously, all the agents may die in the black hole. The issue is that the agents cannot implement the cautious movement strategy outlined in Section \ref{sec:preliminaries}.
This is due to the fact that in order to use this strategy, three agents must be co-located, but we consider a scattered initial configuration here. Thus, each agent does the cautious movement individually. The movement strategy used by each agent is \emph{`Mark and Move' - `Return' - `Delete and Move'}. Basically, when an agent $a_i$ decides to move through a port, it writes this information at its current node $u$ and moves through that port (if the corresponding edge is available). After reaching the new node $v$, if it survives ($v$ is not a black hole), it returns to $u$ (if the corresponding edge is available), removes the marked port information written at $u$, and finally safely moves to the new node $v$, provided the corresponding edge is available. Now, if the visited node $v$ is a black hole, then the information written on the previous node $u$ stays put. This is a warning that can be seen by some other agent that visits $u$. We call this type of movement the \underline{Individual Cautious Movement (ICM)}. The interesting part is that if two agents moved through the same port, say $p$, in different rounds, and have marked at the node $u$ about $p$, then the third agent that visits $u$ understands that port $p$ leads to a black hole. However, a single mark of information is not enough for the second agent to understand the black hole, as it might be the case that the first agent could not return to $u$ due to the unavailability of the corresponding edge. However, information about two marked ports is enough for the third agent to understand the presence of the back hole, since if it were due to the unavailability of the port, then the second agent could not have left through the same port $p$. More specifically, even if the corresponding edge became available only for one round and the second agent left in that round to $v$, then the first agent must be able to come back from $v$ in that round, and therefore only one marked information should be there at $u$ which is due to the second agent.

The above ICM strategy ensures that at most $2\delta_{BH}$ many agents can die in the black hole. Each individual agent runs its own DFS traversal unless it is stuck by a missing edge. Proceeding in this manner, let us suppose that $2\delta_{BH}$ many agents die in the black hole. Since each agent performs ICM, there are markings at each of the $\delta_{BH}$ nodes adjacent to the black hole. Thus, the problem is solved if at least one agent visits one of these $\delta_{BH}$ nodes. However, since there are only 17 agents left in the graph and none of them can die, the movement of these agents can only be blocked by the missing edge. 
Whenever a group of at least 9 agents is formed, they initiate the algorithm for the rooted case as described in \cite{BHS_gen} and outlined in Section \ref{sec:preliminaries}. A missing edge $(u,v)$ can block at most 16 agents, 8 agents at each of the nodes $u$ and $v$, while ensuring that no group of 9 agents is formed. Hence, in this case, at least one agent proceeds according to its DFS traversal in each round. Thus, it eventually visits one of the $\delta_{BH}$ nodes that have some sort of marked information by the two dead agents, leading to the solution of the 1-BHS problem.     

Now, our objective is to limit the storage at each node to $O(\log n)$ bits. Thus, information corresponding to only a constant number of agents can be stored at each node. We utilize the same idea as mentioned above to search for the black hole. The difference is to restrict the number of DFS traversals being run by the agents on the graph. This is done in the following manner. When an agent visits a node, it checks the ID corresponding to the information written on the whiteboard. If no such information is present, then it continues its traversal. However, if it finds that some smaller ID agent has already visited the current node, then it dismisses its own DFS traversal and follows the path followed by this smaller ID agent. But if the information corresponds to some larger ID agent than its own ID, then it continues its traversal by overwriting the information on the whiteboard. By proceeding according to this, either a group of 6 agents is formed eventually or at least one agent reaches one of the $\delta_{BH}$ nodes. In both scenarios, we have our solution to the problem of 1-BHS. This explains the overall idea of our algorithm. 

\subsection{The Algorithm: \textsc{Algo\_Bhs\_Scattered}}
 Before describing the algorithm in detail, we provide a description of the parameters maintained by each agent $a_i$.
\begin{itemize}
    \item $a_i.state$: This is a binary variable. It can take values either $explore$ or $backtrack$. It is initialized to $explore$.
    \item $a_i.success$: This is a binary variable that can take values either $True$ or $False$. This is required to identify whether the movement of $a_i$ in the previous round was successful or not. It is initially set to $True$.
    \item $a_i.pout$: It stores the port that will be used by the agent to exit from the node in the current round. It can take values from $\{-1, 0, 1, \ldots, \delta_v-1\}$, where $\delta_v$ is the degree of the current node $v$. This parameter is initially set to $-1$.
    \item $a_i.pin$: It stores the port used by the agent to enter into the node at the beginning of the current round. It can take values from $\{-1,0, \ldots, \delta_v-1\}$ where $\delta_v$ is the degree of the current node $v$. It is initially set to $-1$. 
    \item $a_i.ID$: This variable stores the ID of the agent.

    \item $a_i.grp$: If an agent is part of a group of at least 9 agents, it updates $a_i.grp$ to $True$. It is initialized to $False$.
     \item $a_i.grp\_ID$: If an agent is part of a group of at least 9 agents, it updates $a_i.grp\_ID$ to the smallest ID agent present in the group. It is initialized to $\bot$.

\end{itemize}

While traversing through the graph, agent $a_i$ also updates information on the whiteboard. The parameters stored on the whiteboard by the agent $a_i$ are as follows:
\begin{itemize}
    \item $wb_v(a_i.parent)$: This parameter stores the port at $v$ from which $a_i$ entered node $v$ for the first time while running its DFS. This can take value from $\{-1, 0, 1, \ldots, \delta_v-1\}$.
    
    \item $wb_v(a_i.recent)$: This parameter stores the last port at $v$ used by the agent $a_i$ to continue its traversal. It enables other agents, such as $a_j$, to track the traversal path of agent $a_i$.

    The information $wb_v(a_i.parent)$ and $wb_v(a_i.recent)$ together constitutes the DFS information corresponding to the agent $a_i$. For the sake of simplicity, we call this information the \underline{travel information}.
    
    \item $wb_v(marked_1)=(a_i.pout, a_i.ID)$: The agent $a_i$ writes its ID and its computed value of $pout$ that it uses to exit from node $v$. Initially, it is set to $\bot$. We denote the agent $a_i$ whose information is written at $wb_v(marked_1)$ with $A_1$.
    \item $wb_v(marked_2)=(a_i.pout, a_i.ID)$: The agent $a_i$ writes its ID and computed value of $pout$ that it uses to exit from $v$. When there is already some information written at $wb_v(marked_1)$, and some other agent wants to either move through the same port or any other port, it writes its information at $marked_2$. Note that the values of $recent$ port and $marked$ port may be the same for an agent $a_i$. However, we use these two parameters for two separate purposes in the algorithm. We denote the agent $a_i$ whose information is written at $wb_v(marked_2)$ with $A_2$.

    For the sake of simplicity, we call this information the \underline{marked port information}.
    \item $wb_v(grp)$: This parameter can take values either $True$ or $False$. It helps to identify whether a group of at least 9 agents is formed or not. Initially, it is set to $False$. Once the value is set to $True$, it never changes to $False$.
    \item $wb_v(grp\_ID)$: If a group consisting of at least 9 agents is formed, then whenever any agent $a_i$ from this group visits node $v$, it updates this parameter with the value of its $a_i.grp\_ID$.  
\end{itemize}

Let $2\delta_{BH}+17$ many agents be initially positioned arbitrarily at the safe nodes of $G$. If there is at least one group of at least $9$ agents in the initial configuration, then they can proceed with the rooted algorithm from \cite{BHS_gen}. The problem of 1-BHS is solved in this scenario. Hence, let us assume that there is no group of at least $9$ agents in the initial configuration. 

Our algorithm operates in alternating odd and even rounds. In odd rounds, agents perform a specific set of actions, while in even rounds, they execute a different set of operations. Let an agent $a_i$ be present at a node $v$ in round $t$. If $t\mod 2=1$, then $a_i$ updates its $a_i.success=True$ if its attempt to move in the round $t-1$ was successful. Otherwise, if its attempt to move in the round $t-1$ was unsuccessful, then it sets $a_i.success=False$. On the other hand, if $t\mod 2=0$, then  $a_i$ performs its movement based on the various cases as described below. Before proceeding with the specific cases, we describe ICM performed by  $a_i$.

\vspace{0.1cm}
\noindent\textbf{Individual Cautious Movement (ICM):} Let $a_i$ be present at $v$. Let us suppose that $a_i$ decides to exit from $v$ by performing ICM via port $p_i$ at round $t$ to reach the adjacent node, say $v'$. It writes its marked port information at $v$ and moves through the port $p_i$ at the end of round $t$. In case the move fails due to a missing edge, it removes the marked port information in the subsequent odd round and starts ICM freshly from the next even round. After successfully arriving at the adjacent node $v'$ at some round $t'\geq t$, $a_i$ (if alive) returns to $v$ and deletes its marked port information at $v$. After deleting the marked port information, $a_i$ finally moves through port $p_i$ to reach the adjacent node safely. Once $a_i$ reaches $v'$ safely, it is said to complete ICM. 
Now we are ready to provide the details of our algorithm.

 We first describe the case when an agent $a_i$ is not yet a part of a group of at least 9 agents, i.e., $a_i.grp=False$. In this case, we have the following four cases. 

\vspace{0.2cm}
\noindent (A) If $a_i$ is at node $v$ in $explore$ state after completing its ICM and $a_i.success=True$. 
In this scenario, we have the following two cases:

\begin{itemize}
    \item[A.1] \textit{There are no marked port information at the current node $v$}- This implies that the values of both $marked_1$ and $marked_2$ are $\bot$. Agents check if there is any travel information corresponding to $a_j$ such that $a_j.ID$ is smaller than all the agents present at the current node. If information corresponding to such $a_j$ is present, then $a_i$ moves through the port $a_j.recent$, where this value is extracted from the whiteboard. Otherwise, if there is no such travel information present at $v$, then $a_i$ proceeds as follows. If $a_i.ID$ is the smallest among all the agents at $v$, then $a_i$ proceeds with its own DFS traversal via ICM.\footnote{Note that by ``proceeding with DFS via ICM," we mean that the agent writes its traversal information at the current node according to the DFS traversal of $a_i$. Moreover, if the agent needs to explore from the current node, the use of ICM is required; however, if it needs to backtrack, ICM is not necessary as it is definite that the node is a safe node.} On the other hand, if $a_i.ID$ is not the smallest among all the agents at $v$\footnote{The comparison among agent IDs is performed only among those agents that are present at the current node after completing their ICM.}, then $a_i$ waits at the current node.
    
    \item[A.2] \textit{There is one marked port information at the current node $v$}- Without loss of generality, let $marked_1 \neq \bot$ and the information is written there corresponding to an agent $A_1$. If $A_1$ is not present at $v$, then $a_i$ proceeds based on the following cases. If $a_i.ID$ is smaller than all the agents at the current node, as well as smaller than $A_1.ID$, then $a_i$ continues its DFS via ICM. If $a_i.ID$ is the smallest among all the agents at the current node, but $a_i.ID>A_1.ID$, then $a_i$ sets $a_i.pout=A_1.pout$ (where $A_1.pout$ is extracted from $wb_v(marked_1)$), and attempts to move through this port by writing its information in $wb_v(marked_2)$. If $a_i$ is not the smallest ID agent at the current node, then it waits at the current node. 
    
    On the other hand, if $A_1$ is present at $v$ after verifying that the adjacent node it had to visit is a safe node (i.e., $A_1.success=True$ and $A_1$ has entered into $v$ via $A_1.recent$ port), then $a_i$ proceeds based on the following. If $A_1$ is the smallest ID agent from all the agents at $v$ and the ID of the agent whose travel information is written at $v$ (if any), then $a_i$ sets $a_i.pout=A_1.pin$ and moves through $a_i.pout$. Otherwise, if there is travel information at $v$ corresponding to an agent $a_j$ such that $a_j.ID$ is smaller than all the agents at $v$, then $a_i$ sets $a_i.pout=a_j.recent$ where the value of $a_j.recent$ is extracted from $wb_v$. Agent $a_i$ moves through its computed value of $a_i.pout$. If no such travel information corresponding to $a_j$ is found at $v$, then if $a_i$ is the smallest ID agent at $v$, it proceeds according to its DFS via ICM. Else, $a_i$ waits at the current node.

    \item[A.3] \textit{Both the marked port information are present at $v$}- This implies both $marked_1 \neq \bot$ and $marked_2 \neq \bot$. Let the information written in $marked_1$ correspond to $A_1$, and the information in $marked_2$ correspond to $A_2$. Here, we have three cases.
    \begin{itemize}
        \item If neither of $A_1$ and $A_2$ is present at $v$ (to delete their respective $marked$ port information), then $a_i$ checks if $marked_1=marked_2$, then black hole is found. Otherwise, if $marked_1\neq marked_2$, then $a_i$ waits at the current node. 
        \item If one of $A_1$ and $A_2$ is present at $v$ to delete its marked port information, then $a_i$ proceeds according to the following. Without loss of generality, let $A_1$ be present at $v$. If $A_1.ID$ is the minimum among all the agents at $v$ and the travel information present at $v$, then $a_i$ moves along with $A_1$ by setting $a_i.pout=A_1.pin$. If $A_1.ID>A_2.ID$ and $A_2.ID$ is smaller than all the agents present at the current node, then $a_i$ proceeds based on the following two cases, (a) If $a_i$ is the smallest ID agent among all the agents at the current node, then $a_i$ sets $a_i.pout$ same as the $marked_2$ (corresponding to $A_2$). Agent $a_i$ then writes its information in $marked_1$ and attempts to move through its port $a_i.pout$ by performing ICM; (b) If $a_i$ is not the smallest ID agent among all the agents at the current node, then $a_i$ waits at the current node $v$.

        Now, if $A_1.ID>A_2.ID$ and $A_2.ID$ is not the smallest among all the agents at the current node, then if $a_i$ is the smallest ID agent among all the agents at $v$, then it proceeds with its own DFS via ICM by writing its information in $marked_1$. Otherwise, $a_i$ waits at the current node. 

        \item If both $A_1$ and $A_2$ are present at $v$ to delete their marked port information, this means both the $marked$ port information will be deleted by the respective agents. This case thus reduces to case 1 when both the $marked$ port informations are $\bot$. 
    \end{itemize}
\end{itemize}
This completes the algorithm of an agent $a_i$ when it is present at a node in the $explore$ state after completing its ICM. 

\noindent (B) If $a_i$ is at node $v$ in $explore$ state after completing its ICM and $a_i.success=False$. For this scenario, we have the following two cases:
\begin{itemize}
    \item[B.1] \textit{There is no agent $a_j$ with $a_j.success=True$ and $a_j$ has completed its ICM}- This means no new agent has entered into the current node. Thus, $a_i$ continues its attempt to move through its computed $pout$ value.
    \item[B.2] \textit{There is at least one agent $a_j$ present at $v$ with $a_j.success=True$ and $a_j$ has completed its ICM}- In this case, if $a_i$ has the minimum ID among all agents at $v$, then it re-attempts to move through its $a_i.pout$. Otherwise, it waits at $v$.       
\end{itemize}
Note that, in this case, $marked$ port information checking is not required due to the fact that an agent first arrives at $v$ with $success=True$ and then writes its $marked$ port information at $v$. If there was already any $marked$ port information at $v$, then the decision was taken by $a_i$ based on that information according to the case (A) above. 

\vspace{0.2cm}
\noindent (C) If $a_i$ is at node $v$ with $a_i.state=backtrack$ and $a_i.success=True$. For this scenario, we have the following two cases.
    \begin{itemize}
        \item[C.1] There is no information related to $marked$ ports at the current node $v$- In this case, $a_i$ checks if there is travel information corresponding to an agent $a_j$ such that $a_j.ID<$ IDs of all agents at $v$. If such travel information is present at $v$, then $a_i$ sets $a_i.pout=a_j.recent$ where the value of $a_j.recent$ is extracted from $wb_v$. If such travel information is not present, then $a_i$ compares its ID with the IDs of other agents present at the current node. If $a_i.ID$ is the smallest among all the agents present at $v$, then $a_i$ continues its DFS via ICM. Otherwise, if $a_i$ is not the smallest ID agent present at $v$, then $a_i$ waits at the current node.  
        \item[C.2] There is one information related to $marked$ port at $v$- Without loss of generality, let the information be present in $marked_1$ and it corresponds to the agent $A_1$. In this scenario, we have two cases.
            \begin{itemize}
            \item If $A_1$ is present at $v$ to delete its marked port information at $v$, then $a_i$ checks if it is the smallest ID agent among all the agents at $v$. If yes, then $a_i$ continues its own DFS. Otherwise, if $a_i$ is not the smallest ID agent among all the agents at $v$, then it waits at the current node.
            \item If $A_1$ is not present at $v$, then $a_i$ checks if it is the smallest ID agent among all agents at $v$ and smaller than $A_1.ID$ as well, then $a_i$ continues its DFS via ICM. Otherwise, if $a_i.ID>A_1.ID$ and $a_i.ID$ is smaller than all agents at $v$, then $a_i$ stops its DFS, sets $a_i.pout=A_1.pout$ (value of $A_1.pout$ is extracted from $wb_v(marked_1)$), and attempts to move through the port $a_i.pout$ by performing ICM and writing its information in $marked_2$. Otherwise, if $a_i$ is not the smallest ID agent, then $a_i$ waits at $v$. 
            \end{itemize}
        \item[C.3] Both the $marked$ port information are present at $v$- Let the information of $marked_1$ correspond to $A_1$ and the information of $marked_2$ correspond to $A_2$. We have the following three subcases.
            \begin{itemize}
            \item If neither of $A_1$ or $A_2$ is present at $v$ to delete their marked port information, then $a_i$ waits at $v$ if $marked_1 \neq marked_2$. If $marked_1=marked_2$, then black hole is found.
            \item If one of $A_1$ and $A_2$ is present at $v$ to delete its marked port information, then $a_i$ decides based on the following scenarios. Without loss of generality, let $A_1$ be present at $v$. If $A_1.ID$ is the minimum from all agents at $v$, including the IDs of the agents written at travel information at $v$ (if any), then $a_i$ sets $a_i.pout=A_1.pout$ (where $A_1.pout$ is extracted from $wb_v(marked_1)$) and moves along with $A_1$. Otherwise, if $A_2.ID$ is the smallest among all agents at $v$, then $a_i$ moves through $A_2.pout$ (where the value of $A_2.pout$ is extracted from $wb_v(marked_2)$) by writing its information in $marked_1$ if $a_i$ is the smallest ID agent among all the agents at $v$. If $a_i$ is not the smallest ID agent at $v$, then it waits at the current node. The final case is if there is another agent $a_j$ at $v$ after completing its ICM that is smaller than all the agents at $v$, including the travel information, then $a_i$ stops the execution of its DFS and waits at the current node $v$. If $a_i$ is the smallest ID agent at $v$ then it performs its DFS via ICM by writing the marked port information at $marked_1$.
            \item If both $A_1$ and $A_2$ are present at $v$ to delete their marked port information, and either of them is the minimum from all the agents present at the current node, as well as the travel information, then $a_i$ moves along with it. Otherwise, if $a_i$ is the smallest ID agent present at $v$, then it continues its DFS via ICM. Otherwise, $a_i$ waits at the current node. 
            \end{itemize}
\end{itemize}

\noindent (D) If $a_i$ is at node $v$ with $a_i.success=False$ and $a_i.state=backtrack$: In this case, $a_i$ continues to execute its own DFS if there is no new agent with $success=True$ after completing its ICM present at $v$ with ID smaller than that of $a_i$. However, if there is such an agent present, then $a_i$ waits at the current node.

\vspace{0.3cm}\noindent (II) Algorithm of agent $a_i$ when it is part of a group with at least 9 agents ($a_i.grp=True$): At a round, say $t$, during the execution of our algorithm, at least $9$ agents (say $x$ many agents) are present at a node $v$. From these $x$ agents, let $p$ many agents be such that they have to return to their previous nodes to delete their $marked$ port information on their respective nodes (i.e., they have not yet completed their ICM). If $x-p\geq 9$, then there is a group of at least $9$ agents that have completed their ICM. These agents initiate the Rooted Algorithm as described in \cite{BHS_gen}. Each agent within this group sets $a_i.grp=True$ and, while writing its travel information at any visited node $v$, it includes this information on the whiteboard by writing $wb_v(grp)=True$. When another agent, say $a_j$ with $a_j.grp=False$ (indicating that it has not started the rooted algorithm and thus is not a part of any group that is executing the rooted algorithm) visits such a node where information is written corresponding to an agent that has $grp=True$ value, $a_j$ understands that a group of 9 agents has initiated the rooted algorithm. The inclusion of this information serves an important purpose: if an agent that is either alone or part of a group with less than $9$ agents encounters a node containing this information (travel information corresponding to an agent $a_i$ with $a_i.grp=True$), it terminates its execution of the algorithm, as it is definite that the group of at least $9$ agents will solve the 1-BHS problem. It may be the case that more than one group of 9 agents are formed that have initiated the rooted algorithm to solve the problem of 1-BHS. Since each node has $O(\log n)$ bits of memory, the information corresponding to each of these groups cannot be stored at the whiteboard. Thus, to resolve this issue, when a group of 9 agents is formed, each agent $a_i$ that is a part of the group stores the ID of the agent that has the smallest ID among the group in the parameter $a_i.grp\_ID$. When these agents update their information on the whiteboard (as per the rooted algorithm of 1-BHS), they also update the value of $grp\_ID$ on the whiteboard of the visited node. When a group (say $G_1$) visits a node where information corresponding to another group (say $G_2$) is written, such that $grp\_ID$ of the agents in $G_1$ is less than the $grp\_ID$ of the agents in $G_2$, then $G_2$ stops its execution of its algorithm. 

Now suppose that $x-p<9$. In this case, all agents at the current node understand that a group of 9 agents can be formed. Thus, these $x-p$ agents remain at current node. The remaining $p$ agents move to their previous nodes (say at round $t_1$) in order to delete their respective $marked$ port information. Since at most one edge can disappear, at most one agent from these $p$ many agents may not be able to move from the current node to its previous node to delete its $marked$ port information. Even if all the $p$ agents successfully move to their previous nodes, at most one agent may not be able to return to the current node. Thus, the $x-p$ agents wait at the current node unless a group of $9$ agents is formed at the current node, consisting of all the agents that have completed their ICM. This completes the description of \textsc{Algo\_Bhs\_Scattered}. The pseudo-codes are added in the Appendix B. Let's proceed with the correctness and analysis of \textsc{Algo\_Bhs\_Scattered}.

\subsection{Correctness and Analysis}
In this section, we show the correctness of \textsc{Algo\_Bhs\_Scattered}. Let the total number of agents in the graph be $l$, i.e., $l=2\delta_{BH}+17$. In the analysis, we denote the $i^{th}$ smallest ID agent with $\alpha_i$. Hence, $\alpha_1.ID<\alpha_2.ID<\ldots<\alpha_l.ID$. As per our algorithm, if agent $a_i$ meets $a_j$ or it sees the information written by some agent $a_j$, and $a_i.ID>a_j.ID$, then agent $a_i$ stops its own DFS and replicates the movement of agent $a_j$. Based on this, we have the following observations. 

\begin{lemma}\label{lem:atmost_travel_info}
    \textsc{Algo\_Bhs\_Scattered} never instructs more than one agent to write travel information on the whiteboard of any node at any round.
\end{lemma}
\begin{proof}
If there is no travel information written at a node, say $v$, at the beginning of a round $t$, then in round $t$, only the agent with the smallest ID among all agents currently present at $v$ is instructed to write its travel information on the whiteboard, as per our algorithm. If travel information corresponding to some agent $a_j$ is already written at $v$, then any agent visiting $v$ with an ID smaller than that of $a_j$ and also being the agent with the smallest ID among those currently at $v$ will overwrite the existing travel information. Otherwise, the agent replicates the movement of $a_j$, and the travel information at $v$ remains unchanged. Hence, the statement follows.
\end{proof}

\begin{lemma}\label{lem:atmost_marked_info}
As per \textsc{Algo\_Bhs\_Scattered}, an agent writes its marked port information on the whiteboard only if at least one of its marked ports is $\bot$.
\end{lemma}
\begin{proof}

If there is no marked port information at a node $v$, and more than one agents are present at $v$, then exactly one of the agents writes the marked port information at the whiteboard. This is because only the smallest ID agent from all the agents at the current node proceeds as per its DFS algorithm. The remaining agents replicate the movement of the smallest ID agent cautiously (refer to cases A.1 and C.1).

 Now consider the case where one marked port information already exists at node $v$, and multiple agents are present. Let the existing information correspond to agent $a_i$. If $a_i$ is present at $v$ and has verified that the adjacent node through the marked port is safe, then $a_i$ removes its marked port information from $v$. Consequently, in the current round $t$, the agent with the smallest ID among those at $v$ writes its own marked port information, ensuring that only one such entry remains at $v$ at the end of round $t$. On the other hand, if $a_i$ is not present at $v$, then its marked port information remains intact. In this case, the agent with the smallest ID at $v$ either continues its DFS or replicates the movement of $a_i$, and writes its own marked port information. As a result, $v$ contains two marked port entries at the end of round $t$ (refer to cases A.2 and C.2). 

Finally, consider the case where two marked port information are already present at $v$. If at least one of the agents corresponding to these marked port information is present at $v$, then at least one of the two marked port information is deleted. On the other hand, if none of the agents whose marked port information is written at $v$ are present, then no agent proceeds further; instead, all agents at $v$ wait at the current node (refer to cases A.3 and C.3). Hence, our statement follows.
\end{proof}
    
\begin{observation}\label{obs:towardsBH}
    Suppose a node $v$ has non-empty information on both of its marked ports at round $t$. If, in the next even round following $t$, none of the agents associated with these information return, then it must be the case that one of the two marked ports at $v$ leads to the black hole.
\end{observation}

\begin{lemma}\label{obs:min_move}
    The agent $\alpha_1$ never deviates from its original DFS path unless it becomes part of a group comprising at least 9 agents. The agent $\alpha_1$ may encounter one of the following three scenarios: (i) it dies upon entering the black hole, (ii) it gets stuck at a node due to a missing edge, or (iii) it waits at a node if both marked ports at that node are different from $\bot$.
\end{lemma}

\begin{proof}
    Since $\alpha_1$ is the agent with the smallest ID, it never follows the DFS traversal of any other agent. Therefore, it continues its own DFS traversal unless it becomes part of a group of 9 agents. During its traversal, $\alpha_1$ may enter the black hole, leading to case (i). Alternatively, it may get stuck at a node due to a missing edge, which corresponds to case (ii). The final possibility arises when $\alpha_1$ is at a node where the whiteboard contains information about both marked ports, each corresponding to another agent. By Lemma \ref{lem:atmost_marked_info}, Lemma \ref{lem:atmost_travel_info}, and Observation \ref{obs:towardsBH}, at most two marked port entries can be present on any node's whiteboard. When both entries are occupied, $\alpha_1$ cannot record its own marked port information and thus waits at the current node, satisfying case (iii). Note that even if $\alpha_1$ is later able to record its own information at such a node in a subsequent round, it continues its DFS traversal without deviation from its original DFS path. This completes the proof.
\end{proof}

\begin{observation}\label{lem:atmost16}
    If the movement of more than 16 agents is blocked due to a missing edge, then a group of 9 agents is definitely formed.
\end{observation}

\begin{remark}
    By \underline{completing an agent's movement}, we mean that either the agent died in the black hole or the agent determined the location of the black hole by reaching a neighbor $v$ of the black hole and outputting the correct port that leads from $v$ to the black hole. 
\end{remark}

\begin{lemma}\label{lem:12lm}
    An agent either dies in the black hole or finds the black hole by performing at most $12lm$ successful movements. 
\end{lemma}
\begin{proof}
An agent requires at most $4m$ successful movements to complete its DFS traversal. Each movement in the DFS is executed using a cautious movement strategy, which consists of 3 individual movements. Therefore, an agent completes its DFS traversal, including the cautious movement strategy, in at most $12m$ successful movements. According to our algorithm, an agent $\alpha_i$ may need to perform its traversal up to $i$ times. This is because there are $i-1$ agents with smaller IDs than $\alpha_i$, each executing their own DFS traversals. In the worst case, agent $\alpha_i$ may sequentially follow the traversals of all these preceding agents before eventually following $\alpha_1$. By Observation~\ref{obs:min_move}, $\alpha_1$ never follows any other agent. Consequently, in the worst case, agent $\alpha_l$ may require at most $12m \cdot l$ successful movements to complete its traversal. This concludes the proof of the lemma.
\end{proof}

\begin{remark}\label{rem:24lm}
    Since agents perform their movements only during even rounds in our algorithm, an agent either dies in the black hole or identifies the black hole within at most $24lm$ rounds, where each round corresponds to a successful movement attempt.
\end{remark}

\begin{theorem}\label{lm:time+correctness}
    \textsc{Algo\_Bhs\_Scattered} solves 1-BHS in $O(m^2)$ rounds.
\end{theorem}
\begin{proof}
During the execution of our algorithm, there are two possible scenarios for the agents: (i) no such group is formed throughout the execution, or (ii) at least one group of 9 agents is formed.

\vspace{0.15cm}
\noindent \textbf{Case (i)}: Let us suppose that a group of 9 agents is not formed during the execution of the algorithm. In this scenario, we claim that 1-BHS is solved in $152m\delta_{BH}$ rounds. In other words, if a group of 9 agents is not formed in $152m\delta_{BH}$ many rounds, the black hole is found. To prove the claim, we guarantee that at least $2\delta_{BH}+1$ many agents complete their movement by $152m\delta_{BH}$ rounds. As per Lemma \ref{lem:12lm}, an agent can have at most $12lm$ successful movements. Hence, in the worst case, cumulatively $(2\delta_{BH}+1)\cdot 12lm$ many successful movements are needed for $2\delta_{BH}+1$ many agents to either die in the black hole or to locate the black hole. Since the remaining 16 agents may not complete their movement, this implies at most $16(12lm-1)$ many successful movements cumulatively can be achieved by these 16 agents. Hence, the cumulative successful movement by all the agents must be at most $(2\delta_{BH}+1)\cdot12lm+16\cdot(12lm-1)$. However, the adversary can make an edge disappear to delay our algorithm. As per Observation \ref{lem:atmost16}, a missing edge can stop at most 16 agents. Hence, at every even round, $2\delta_{BH}+17-16=2\delta_{BH}+1$ successful movements must happen. Let $T$ be the required number of even rounds in the worst case to complete the traversal of $2\delta_{BH}+1$ many agents. Hence, $T=\frac{(2\delta_{BH}+1)\cdot12lm+16\cdot(12lm-1)}{2\delta_{BH}+1}\leq 12lm+\frac{16(12lm)}{3}=76lm$. Thus, if no group of 9 agents is formed, then 1-BHS is solved in $152m\delta_{BH}$ rounds. A black hole is located in this scenario when $2\delta_{BH}$ agents die in the black hole by marking the neighbors of the black hole node, and one agent visits one of these marked neighbors.

\vspace{0.15cm}
\noindent \textbf{Case (ii)}: When a group of 9 agents is formed, the agents initiate the algorithm proposed in \cite{BHS_gen}, which completes in $O(m^2)$ rounds. If this group encounters information about any other group of 9 agents that has also initiated the algorithm of \cite{BHS_gen}, a comparison of the corresponding $grp\_ID$ values is performed. The group with the smaller $grp\_ID$ continues the execution while disregarding any information written by the group with the larger $grp\_ID$. Conversely, the group with the larger $grp\_ID$ terminates its execution. This mechanism ensures that the group with the smallest $grp\_ID$ value ultimately locates the black hole. 
Hence, in the worst case, agents require $O(m\delta_{BH}+m^2)=O(m^2)$ rounds to solve the problem of 1-BHS. 
\end{proof}

\begin{lemma}\label{lm:wbmemory}
    Each node of $G$ is required $O(\log n)$ bits storage to run \textsc{Algo\_Bhs\_Scattered}. 
\end{lemma}
\begin{proof}
    As per Lemma \ref{lem:atmost_travel_info} and \ref{lem:atmost_marked_info}, at each node $v$, there are at most two marked port information and one travel information at any round. Each of these information consists of an agent ID along with a port number. Therefore, each node requires no more than $O(\log n)$ bits of storage. In addition, the parameter $wb_v(recent)$ requires $O(\log \Delta)$ bits, $wb_v(grp)$ requires $O(1)$ bits, and $wb_v(grp\_ID)$ requires $O(\log n)$ bits. Since $\Delta \leq n-1$, the overall storage required per node is $O(\log n)$ bits.
\end{proof}

\begin{lemma}\label{lm:agentmemory}
Each agent requires $O(\log n)$ bits of memory to execute \textsc{Algo\_Bhs\_Scattered}.    
\end{lemma}

\begin{proof}
    Each agent $a_i$ maintains the parameters $a_i.state$, $a_i.success$, and $a_i.grp$, each of which requires $O(1)$ bits of memory. Additionally, the parameters $a_i.pout$ and $a_i.pin$ store port numbers and hence require $O(\log \Delta)$ bits of memory, where $\Delta$ is the maximum degree of the graph. The identifiers $a_i.ID$ and $a_i.grp\_ID$ each require $O(\log n)$ bits of memory. Since $\Delta \leq n$, the overall memory required per agent remains bounded by $O(\log n)$ bits. Hence the statement of our lemma holds. 
\end{proof}

\noindent Hence, our main theorem follows from Theorem \ref{lm:time+correctness}, Lemma \ref{lm:wbmemory}, and Lemma \ref{lm:agentmemory}. 
\begin{theorem}\label{th:scattered}
    \textsc{Algo\_Bhs\_Scattered} solves the 1-BHS problem using $2\delta_{BH} + 17$ agents within $O(m^2)$ rounds. It requires $O(\log n)$ bits of storage per node, and each agent is equipped with $O(\log n)$ bits of memory.
\end{theorem}

\section{Emergent Black Hole Search in Static Graphs}
In this section, we present a solution to \textsc{Ebhs} in static graphs using four co-located agents, assuming that nodes have no storage. The solution relies on a perpetual exploration algorithm based on a Universal Exploration Sequence (UXS) \cite{Reingold2008}, which guarantees that a single agent, say $A_i$, visits every node infinitely often. The algorithm requires $O(\log n)$ memory at the agent, and does not use any storage at the nodes. We refer to this procedure as \textsc{Algo\_UXS}.

\medskip
\noindent\textbf{High-level description of \textsc{Algo\_UXS}.}
A UXS provides a deterministic method for a memory-limited agent to explore an unknown graph without using any node storage. Agent $A_i$ follows a fixed sequence of integers that prescribes its moves using only the port through which it entered the current node. More precisely, suppose agent $A_i$ enters a node $v$ of degree $\delta_v$ through port $p$, and the next element of the sequence is $\beta$. Agent $A_i$ then exits through port $(p + \beta) \bmod \delta_v$. Thus, the movement rule depends only on the local degree and the entry port, and is independent of the global port labeling.

For every positive integer $n$, there exists a UXS of length $\mathrm{poly}(n)$ for the family of all connected graphs with at most $n$ nodes, and such a sequence can be computed deterministically by a sequential algorithm \cite{Reingold2008}. Let $\mathcal{S}_n$ be a UXS. Executing \textsc{Algo\_UXS} with $\mathcal{S}_n$ on any connected static graph $G$ with at most $n$ nodes ensures that agent $A_i$ explores all nodes of $G$.
If agents are unaware of the size of the graph, agent $A_i$ executes this process for guessed value $N=2,2^2,2^3\ldots$. For some $\alpha$, $N=2^\alpha\geq n$. This ensures perpetual exploration, as it repeats this procedure. The agent only needs to store its current position in the sequence, its identifier, and the incoming port number. If agent $A_i$ is aware of $n$, then agent $A_i$ can achieve termination, and the exploration time is proportional to the length of the UXS and the memory requirement is $O(\log n)$.

From the above description, observe that in each round the agent $A_i$ moves from its current node to an adjacent node. Based on the direction of the next move with respect to the previously traversed edge, we classify the movement of $A_i$ into two types:

\begin{itemize}
    \item \textit{Forward move}: Suppose agent $A_i$ moves from node $u$ to node $v$ in round $r$. If in round $r+1$ it moves from $v$ to a node $w \neq u$, this move is called a forward move.
    
    \item \textit{Backward move}: Suppose agent $A_i$ moves from node $u$ to node $v$ in round $r$. If in round $r+1$ it moves from $v$ back to node $u$, this move is called a backward move.
\end{itemize}

In the next section, we describe the cautious chain movement, performed by a group of four agents to replicate a single movement of \textsc{Algo\_UXS}. 

\subsection{The Algorithm: \textsc{Algo\_UXS} using Cautious Chain Movement}
We propose a cautious chain movement performed by 4 agents. Our algorithm replicates each edge traversal of \textsc{Algo\_UXS} by 4 agents via the cautious chain movement. This makes sure that at least one alive agent finds the black hole after it emerges . Below we provide the details of the cautious chain movement.\\
\\
Let $a_1$, $a_2$, $a_3$ and $a_4$ be four agents initially co-located at the home base node $\mathcal{H}$ at round $0$. Since the black hole cannot emerge at node $\mathcal{H}$ at round $0$, all agents are initially safe. Each round $r\geq 0$ of \textsc{Algo\_UXS}, we divide into sub-rounds based on the type of move. We divide this into two cases: (i) $r=0$, and (ii) $r\neq 0$. The algorithm for alive agents is as follows. 

\begin{enumerate}
    \item $\bm{r=0:}$ At round $r=0$, the move determined by \textsc{Algo\_UXS} is a forward move (unless $n=1$). Let all agents be initially co-located at the home base node $\mathcal{H}$. Let $p$ be the outgoing port at $\mathcal{H}$ determined by \textsc{Algo\_UXS}, and let $u$ be the node adjacent to $\mathcal{H}$ via port $p$. At round $r=0$, agent $a_1$ and $a_2$ remains at node $\mathcal{H}$, while agents $a_3$ and $a_4$ move to node $u$ through port $p$.

\item $\bm{r \neq 0:}$ For every round $r \neq 0$, the agents occupy two adjacent nodes, say $v_1$ and $v_2$, where agents $a_1$ and $a_2$ are at node $v_1$, and agents $a_3$ and $a_4$ are at node $v_2$. Agents $a_1$ and $a_2$ at node $v_1$ are aware of port $p_1$, which leads to node $v_2$, and agents $a_3$, $a_4$ at node $v_2$ are aware of port $p_2$, which leads to node $v_1$. The next port $p_3$ at node $v_2$ is determined by \textsc{Algo\_UXS} to reach the next node $v_3$. This move can be either a forward move (i.e., $p_3\neq p_2$) or a backward move (i.e., $p_3 = p_2$). The cautious chain movement handles these two cases as follows. 

\vspace{0.15cm}
\noindent\textbf{Case 1 $(p_3 = p_2)$ (backward):}
Agents $a_3$ and $a_4$ at node $v_2$ compute the next port $p_3$ according to \textsc{Algo\_UXS}. Since $p_3 = p_2$, the intended move is a backward to node $v_1$. The cautious chain movement (refer to Figure \ref{fig:backtrack}) is executed as follows in sub-rounds. 

\begin{enumerate}
    \item \textbf{Sub-round 1:} Agent $a_3$ moves from $v_2$ to $v_1$ through port $p_2$.
    
    \item \textbf{Sub-round 2:} Agents $a_1$, $a_2$ at $v_1$:
    \begin{itemize}
        \item If $a_3$ does not arrive, it declares port $p_1$ at $v_1$ lead to a black hole.
        \item Otherwise, agents $a_1$, $a_2$ gets the information of backtrack from agent $a_3$, and move from $v_1$ to $v_2$ through port $p_1$. 
    \end{itemize}
    
    \item \textbf{Sub-round 3:} Agent $a_4$ at $v_2$:
    \begin{itemize}
        \item If $a_1$ and $a_2$ do not arrive, it declares port $p_2$ at $v_2$ lead to a black hole.
        \item Otherwise, agent $a_4$ moves from $v_2$ to $v_1$ through port $p_2$.
    \end{itemize}
    
    \item \textbf{Sub-round 4:} Agent $a_3$ at $v_1$:
    \begin{itemize}
        \item If $a_4$ does not arrive, it declares port $p_1$ at $v_1$ lead to a black hole.
        \item Otherwise, the surviving agents declare the completion of the round $r\geq 1$.
    \end{itemize}
\end{enumerate}

\begin{figure}
    \centering
\includegraphics[width=0.95\linewidth]{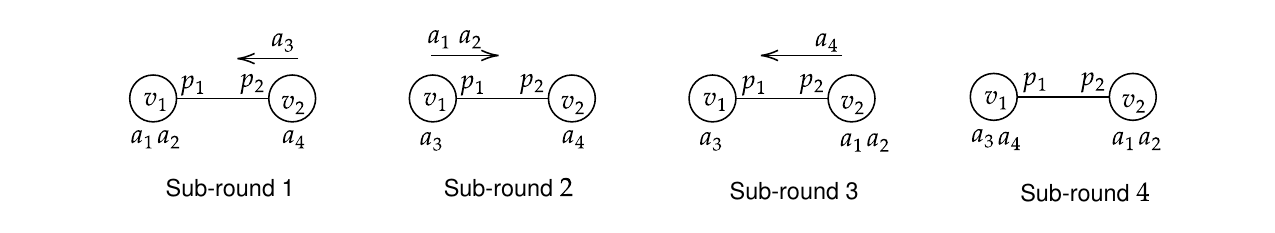}
    \caption{Cautious chain movement for backward move.}
    \label{fig:backtrack}
\end{figure}

\vspace{0.15cm}
\noindent\textbf{Case 2 $(p_3 \neq p_2)$ (forward):}
Agents $a_3$ and $a_4$ at node $v_2$ compute the next port $p_3$ according to \textsc{Algo\_UXS}. Since $p_3 \neq p_2$, the intended move is a forward move to node $v_3$. Let $p_3'$ denote the incoming port at $v_3$. The cautious chain movement (refer to Figure \ref{fig:forward}) is executed as follows in sub-rounds.

\begin{enumerate}
    \item \textbf{Sub-round 1:} Agent $a_3$ moves from $v_2$ to $v_1$ through port $p_2$.
    
    \item \textbf{Sub-round 2:} Agents $a_1$, $a_2$ at $v_1$:
    \begin{itemize}
        \item If $a_3$ does not arrive, it declares port $p_1$ at $v_1$ lead to a black hole.
        \item Otherwise, agents $a_1$, $a_2$ gets the information of forward from agent $a_3$. Agent $a_3$ moves to node $v_2$. 
    \end{itemize}

    \item \textbf{Sub-round 3:} Agent $a_4$ at $v_2$:
    \begin{itemize}
        \item If $a_3$ does not arrive at $v_2$, it declares port $p_2$ at $v_2$ lead to a black hole.
        \item Otherwise, agent $a_4$ moves from $v_2$ to $v_3$ through port $p_3$, and $a_2$ moves from $v_1$ to $v_2$.
    \end{itemize}
    \item \textbf{Sub-round 4:} Agent $a_3$ at $v_2$:
    \begin{itemize}
        \item If $a_2$ does not arrive at $v_2$, it declares port $p_2$ at $v_2$ lead to a black hole.
        \item Otherwise, agent $a_2$ moves from $v_2$ to $v_1$ through port $p_2$.
    \end{itemize}
    \item \textbf{Sub-round 5:} Agent $a_1$ at $v_1$:
    \begin{itemize}
        \item If $a_2$ does not arrive at $v_1$, it declares port $p_1$ at $v_1$ lead to a black hole.
        \item Otherwise, agents  $a_1$, $a_2$ move from $v_1$ to $v_2$ through port $p_1$.
    \end{itemize}

     \item \textbf{Sub-round 6:} Agent $a_3$ at $v_2$:
    \begin{itemize}
        \item If $a_1$, $a_2$ does not arrive at $v_3$, it declares port $p_2$ at $v_2$ lead to a black hole, and moves to node $v_3$ via port $3$.
        \item Otherwise, it moves to node $v_3$ via port $p_3$.
    \end{itemize}

     \item \textbf{Sub-round 7:} Agent $a_4$ at $v_3$:
    \begin{itemize}
        \item If $a_3$ does not arrive at $v_3$, it declares port $p_3'$ at $v_3$ lead to a black hole.
        \item Otherwise, if $a_3$ has information about a black hole, it gets the information about a black hole. Otherwise, the surviving agents declare the completion of the round $r\geq 1$.
    \end{itemize}
\end{enumerate}

\begin{figure}[ht]
    \centering
\includegraphics[width=0.75\linewidth]{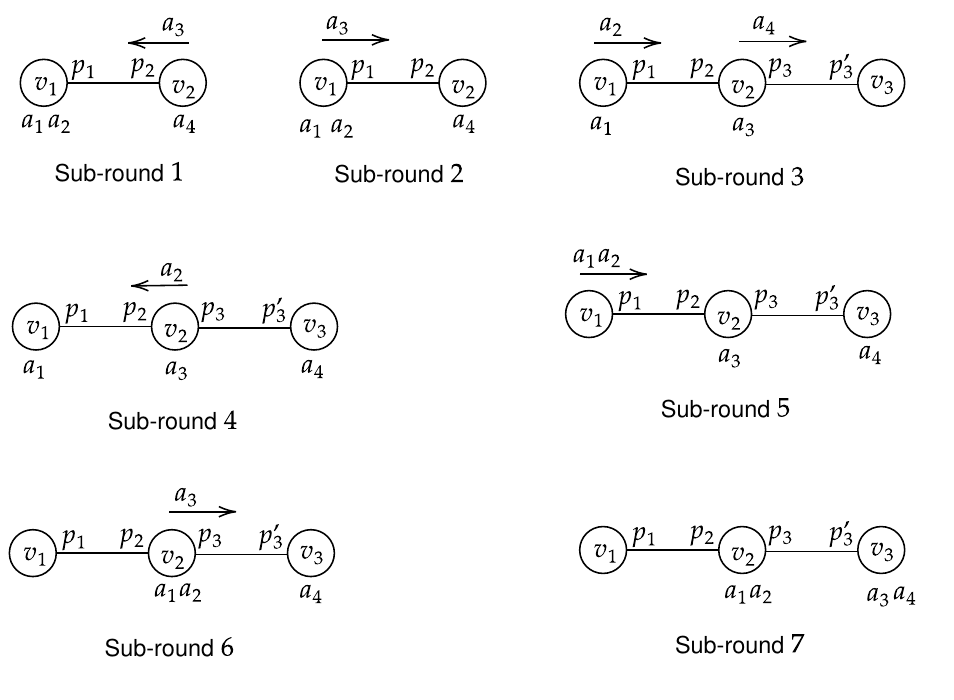}
    \caption{Cautious chain movement for forward move.}
    \label{fig:forward}
\end{figure}

At the beginning of round $r+1$, the agents occupy two adjacent nodes in the same configuration as at the beginning of round $r$. Hence, the same cautious chain movement strategy can be applied in round $r+1$.
\end{enumerate}

\subsection{The correctness and analysis of the algorithm}

\begin{theorem}\label{thm:ebhs}
Simulating every move of \textsc{Algo\_UXS} by the proposed chain movement strategy yields a correct solution to the \textsc{Ebhs} problem.
\end{theorem}
\begin{proof}
    At round $r=0$, the home base node $\mathcal{H}$ is guaranteed to be safe; hence, no agent is destroyed at $\mathcal{H}$ in this round. According to our strategy for round $r=0$, agents $a_1$ and $a_2$ remain at node $\mathcal{H}$, while agents $a_2$ and $a_3$ move to node $u$ as determined by \textsc{Algo\_UXS}.

Assume that up to the completion of round $r-1$ of \textsc{Algo\_UXS}, no black hole has appeared. Suppose that a black hole emerges during the execution of round $r$ of \textsc{Algo\_UXS}, i.e., it may appear in any of the sub-rounds used to simulate round $r$. Since no black hole appears before the end of round $r-1$, the agents are, at that point, located at two adjacent nodes $v_1$ and $v_2$, where agents $a_1$ and $a_2$ are at node $v_1$, and agents $a_3$ and $a_4$ are at node $v_2$. Agents $a_1$ and $a_2$ know the outgoing port $p_1$ at $v_1$ leading to $v_2$, while agents $a_3$ and $a_4$ know the incoming port $p_2$ at $v_2$ from $v_1$ and can compute the next outgoing port $p_3$ according to \textsc{Algo\_UXS}. At round $r$, the agents at node $v_2$ determine the next move determined by \textsc{Algo\_UXS}, which is either a forward move (to a node $v_3 \neq v_1$) or a backward (to node $v_1$). We show that, in both cases, if the black hole emerges during the execution of round $r$, then one of the agents correctly survives and identifies the port leading to the black hole at some round $r'$ where $r'\geq r$.

Let $v_{BH}$ be the node at which the black hole emerges during the execution of round $r$. There are two possible cases: (i) $v_{BH}\notin \{v_1,v_2\}$, or (ii) $v_{BH}\in\{v_1,v_2\}$.

\begin{itemize}
    \item Case (i) $(v_{BH}\notin \{v_1,v_2\})$: In some round $r'\geq r$, the forward move from node $v_2$ leads to node $v_3$, where $v_3 = v_{BH}$. In this case, at the end of the execution of round $r'$, agents $a_1$ and $a_2$ are at node $v_2$, and agents $a_3$ and $a_4$ are at node $v_3$. Since the black hole has emerged at node $v_3$, agents $a_2$ and $a_3$ are destroyed at node $v_3$. 

    During the execution of round $r'+1$ according to \textsc{Algo\_UXS}, agent $a_2$ is supposed to visit agents $a_1$ and $a_2$ in sub-round~1. Since agent $a_2$ has been destroyed, it cannot reach agents $a_1$ and $a_2$. Agents $a_1$ and $a_2$, located at node $v_2$, know the outgoing port through which agents $a_3$ and $a_4$ moved. Hence, agents $a_1$ and $a_2$ correctly identify that this port leads to $v_{BH}$.
    \item Case (ii) $(v_{BH}\in \{v_1,v_2\})$: In this case, as per our strategy at the beginning of round $r$, agents $a_1$, $a_2$ are at node $v_1$, and agents $a_3$ and $a_4$ are at node $v_2$. Now, the black hole can emerge at node $v_1$ or $v_2$ in one of the sub-rounds. From node $v_2$, there are two types of movement that can happen at round $r$: (i) backward, and (ii) forward. Based on these cases, we show the correctness.

    \begin{enumerate}
        \item \textbf{Backward move:} If the black hole emerges during any sub-round of round $r$, and the next move as per \textsc{Algo\_UXS} is backward, all surviving agents correctly determine its location for the following reasons.
        \begin{itemize}
            \item \textbf{Sub-round 1:} If $v_1=v_{BH}$, agents $a_1$, $a_2$ die in sub-round 1, and agent $a_3$ also dies at node $v_1$ as it moves to node $v_1$ in sub-round 1. In sub-round 3, agent $a_4$ is waiting for agents $a_1$ and $a_2$ at node $v_2$. Since $a_1$, $a_2$ were destroyed at sub-round 1, agent $a_4$ does not find $a_1$ and $a_2$ at node $v_2$, and correctly concludes that port $p_2$ leads to the black hole.  

            If $v_2=v_{BH}$, agents $a_3$ and $a_4$ die in sub-round 1. In sub-round 2, agents $a_1$ and $a_2$ are waiting for agent $a_3$ at node $v_1$. Since $a_3$ was destroyed at sub-round 1, agents $a_1$ and $a_2$ do not find $a_3$ at node $v_3$, and correctly conclude that port $p_1$ leads to the black hole.
            \item \textbf{Sub-round 2:} If $v_1=v_{BH}$, agent $a_1$, $a_2$ dies in sub-round 2, and agent $a_3$ also dies at node $v_1$ as it moves to node $v_1$ in sub-round 1. In sub-round 3, agent $a_4$ is waiting for agents $a_1$ and $a_2$ at node $v_2$. Since $a_1$, $a_2$ were destroyed at sub-round 2, agent $a_4$ does not find $a_1$ and $a_2$ at node $v_2$, and correctly concludes that port $p_2$ leads to the black hole.  

            If $v_2=v_{BH}$, agent $a_4$ dies in sub-round 2. In sub-round 2, $a_1$ and $a_2$ move node $v_2$, and dies. Since, in sub-round 4, $a_3$ expects $a_4$ at node $v_1$, and it does not find it, and correctly concludes that port $p_1$ leads to the black hole.
            \item \textbf{Sub-round 3:} If $v_1=v_{BH}$, agent $a_3$ and $a_4$ die in sub-round 3, and $a_1$, $a_2$ are alive at node $v_2$, and remain alive at the end of sub-round 4 of round $r$. In the sub-round 1 of round $r+1$, agents $a_1$ and $a_2$ are expecting $a_3$ at node $v_2$, and they do not find it, and correctly conclude that port $p_2$ leads to the black hole.

            If $v_2=v_{BH}$, agent $a_1$, $a_2$ and $a_4$ die in sub-round 3. In sub-round 4, agent $a_3$ is waiting for agent $a_4$ at node $v_1$. Since $a_4$ was destroyed at sub-round 3, agent $a_3$ does not find $a_4$ at node $v_1$, and correctly concludes that port $p_1$ leads to the black hole.
            \item \textbf{Sub-round 4:} If $v_1=v_{BH}$, agent $a_3$ and $a_4$ die in sub-round 4. In the sub-round 1 of round $r+1$, agent $a_1$ is expecting $a_3$ at node $v_2$, and it does not find it, and correctly concludes that port $p_2$ leads to the black hole.

            If $v_2=v_{BH}$, agents $a_1$ and $a_2$ die at sub-round 4 of round $r$. In sub-round 1 of $r+1$, agent $a_3$ moves to node $v_2$, and it dies. In sub-round 3 of $r+1$, agent $a_4$ is expecting agents $a_1$, $a_2$ or $a_3$ at node $v_1$, and it does not find them, and correctly concludes that port $p_1$ leads to the black hole.
        \end{itemize}
        In this way, it correctly identifies a port that leads to a black hole if a black hole emerges in any sub-round of round $r$ for backward. 
        \item \textbf{Forward move:} If the black hole emerges during any sub-round of round $r$, and the next move as per \textsc{Algo\_UXS} is forward, all surviving agents correctly determine its location for the following reasons.
        \begin{itemize}
            \item \textbf{Sub-round 1:} If $v_1=v_{BH}$, agents $a_1$, $a_2$ die in sub-round 1, and agent $a_3$ also dies at node $v_1$ as it moves to node $v_1$ in sub-round 1. In sub-round 3, agent $a_4$ is waiting for agent $a_3$ at node $v_2$. Since $a_3$ was destroyed at sub-round 1, agent $a_4$ does not find $a_3$ at node $v_2$, and correctly concludes that port $p_2$ leads to the black hole.  

            If $v_2=v_{BH}$, agents $a_3$ and $a_4$ die in sub-round 1. In sub-round 2, agents $a_1$ and $a_2$ are waiting for agent $a_3$ at node $v_1$. Since $a_3$ was destroyed at sub-round 1, agents $a_1$ and $a_2$ do not find $a_3$ at node $v_3$, and correctly conclude that port $p_1$ leads to the black hole.
            \item \textbf{Sub-round 2:} If $v_1=v_{BH}$, agent $a_1$, $a_2$ dies in sub-round 2, and agent $a_3$ also dies at node $v_1$ as it moves to node $v_1$ in sub-round 1. In sub-round 3, agent $a_4$ is waiting for agent $a_3$ at node $v_2$. Since $a_3$ was destroyed at sub-round 1, agent $a_4$ does not find $a_3$ at node $v_2$, and correctly concludes that port $p_2$ leads to the black hole.   

            If $v_2=v_{BH}$, agent $a_4$ dies in sub-round 2. In sub-round 2, $a_3$ moves node $v_2$, and dies. In sub-round 3, agent $a_2$ moves node $v_2$, and dies as well. Since in sub-round 5, $a_1$ expects $a_2$ at node $v_1$, but does not find it, it correctly concludes that port $p_1$ leads to the black hole.

            \item \textbf{Sub-round 3:} If $v_1=v_{BH}$, agents $a_1$ and $a_2$ dies in sub-round 3. In sub-round 4, agent $a_3$ expects agent $a_2$ at node $v_2$. Since agent $a_2$ dies at sub-round 3, agent $a_3$ does not find agent $a_2$ at node $v_2$, and correctly concludes port $p_2$ leads to a black hole.

            If $v_2=v_{BH}$, agents $a_3$, $a_4$ dies in sub-round 3. In sub-round 3, agent $a_2$ moves node $v_2$, and dies as well. Since in sub-round 5, $a_1$ expects $a_2$ at node $v_1$, but does not find it, it correctly concludes that port $p_1$ leads to the black hole.

          \item \textbf{Sub-round 4:} If $v_1=v_{BH}$, agents $a_1$ and $a_2$ dies at node $v_1$. In sub-round 6, agent $a_3$ expects agents $a_1$ and $a_2$ at node $v_2$. Since agents $a_1$, $a_2$ die in sub-round 4, agent $a_3$ does not find agents $a_1$ and $a_2$ at node $v_2$, and correctly concludes port $p_2$ leads to a black hole.

            If $v_2=v_{BH}$, agent $a_2$ and $a_3$ dies in sub-round 4 at node $v_2$. Since in sub-round 5, $a_1$ expects $a_2$ at node $v_1$, but does not find it, it correctly concludes that port $p_1$ leads to the black hole.
          
           \item \textbf{Sub-round 5:} If $v_1=v_{BH}$, agents $a_1$ and $a_2$ dies at node $v_1$. In sub-round 6 , agent $a_3$ expects agents $a_1$ and $a_2$ at node $v_2$. Since agents $a_1$, $a_2$ die in sub-round 4, agent $a_3$ does not find agents $a_1$ and $a_2$ at node $v_2$, and correctly concludes port $p_2$ leads to a black hole.

            If $v_2=v_{BH}$, agent $a_2$ dies in sub-round 5 at node $v_2$. In sub-round 6, agents $a_1$ and $a_2$ also die at node $v_2$. In sub-round 7, agent $a_4$ expects agent $a_3$ at node $v_3$. Since agent $a_2$ dies in sub-round 5, agent $a_4$ does not find it in sub-round 7, and correctly concludes port $p_3'$ leads to a black hole.

            \item \textbf{Sub-round 6:} If $v_1=v_{BH}$, no agent dies in this case as agents $a_1$, $a_2$ and $a_3$ are at node $v_2$, and agent $a_4$ is at $v_3$. In further sub-rounds of $r$, since
no agent visits node $v_1$, no agent dies during the sub-rounds of $r$. In some $r'> r$, it will be identified when agents visit node $v_1$ as a forward move (identical to Case (i)).

            If $v_2=v_{BH}$, agents $a_1$, $a_2$, $a_3$ die in sub-round 6 at node $v_2$. In sub-round 7, agent $a_4$ expects agent $a_3$ at node $v_3$. Since agent $a_2$ dies in sub-round 6, agent $a_4$ does not find it in sub-round 7, and correctly concludes port $p_3'$ leads to a black hole.

            \item \textbf{Sub-round 7:} If $v_1=v_{BH}$, no agent dies in this case as agents $a_1$, $a_2$ and $a_3$ are at node $v_2$, and agent $a_4$ is at $v_3$. In further sub-rounds of $r$, since
no agent visits node $v_1$, no agent dies during the sub-rounds of $r$. In some $r'> r$, it will be identified when agents visit node $v_1$ as a forward move (identical to Case (i)).

            If $v_2=v_{BH}$, agents $a_1$, $a_2$ die in sub-round 6 at node $v_2$. Agents $a_3$ and $a_4$ are at node $v_3$. In sub-round 1 of $r+1$, agent $a_3$ moves to node $v_2$, and it dies. In sub-round 3 of $r+1$, agent $a_4$ is expecting agent $a_3$ at node $v_3$, and it does not find them, and correctly concludes that port $p_3'$ leads to the black hole.
        \end{itemize}
        In this way, it correctly identifies a port that leads to a black hole if a black hole emerges in any sub-round of round $r'$ for backward.
    \end{enumerate} 
\end{itemize}

Since \textsc{Algo\_UXS} ensures perpetual exploration, the above analysis implies that if a black hole emerges at a round $r'$, its location is detected by at least one surviving agent. Furthermore, this information is eventually known to all surviving agents. This completes the proof.
\end{proof}

Based on Theorem \ref{thm:ebhs}, we have the following remarks.

\begin{remark}\label{re:1}
    If the black hole emerges in $\mathrm{poly}(n)$ time, our algorithm solves \textsc{Ebhs} in $\mathrm{poly}(n)$ time, and each agent needs $O(\log n)$ memory.
\end{remark}

\begin{remark}\label{re:2}
(\textsc{Ebhs} with knowledge of $n$) If the agents know the value of $n$, then after the black hole emerges, they require time proportional to the length of the UXS and $O(\log n)$ memory. In this case, instead of executing \textsc{Algo\_UXS} for $N = 2, 2^2, \ldots$, the agents repeat the UXS corresponding to $n$ indefinitely. Each agent only needs to store a constant number of port numbers, the IDs of the agents $a_1$, $a_2$, and $a_3$, and the current position in the UXS. Therefore, $O(\log n)$ bits of memory are sufficient.
\end{remark}



\begin{remark}\label{re:4}
(\textsc{Ebhs} using any single agent exploration algorithm $\mathcal{A}$: a tradeoff between time and memory/storage requirement) Consider any single agent exploration algorithm $\mathcal{A}$. It is easy to observe that in algorithm $\mathcal{A}$, every move performed by an agent is either a forward move or a backward. Such algorithms include most standard exploration procedures in the literature. By replacing each move of $\mathcal{A}$ with the cautious chain movement of four agents, we obtain a solution to \textsc{Ebhs} as follows.

If $\mathcal{A}$ performs perpetual exploration, each movement of a single agent in $\mathcal{A}$ is replaced by a cautious chain movement. Otherwise (exploration algorithm $\mathcal{A}$ with termination), each movement of a single agent in $\mathcal{A}$ is replaced by a cautious chain movement, and instead of termination, restart $\mathcal{A}$. When node storage is used, a constant amount of additional storage suffices to distinguish between old and new information.

The simulation preserves the node storage and prior knowledge required by $\mathcal{A}$. Each agent only needs to store its role in the chain movement, the identifiers of the other agents, and the local information used by $\mathcal{A}$; this requires $O(\max\{\mu,\log n\})$ bits of memory per agent, where $\mu$ is the memory used by $\mathcal{A}$. The time overhead is only a constant factor per simulated move. Hence, any exploration algorithm with node storage $\sigma$, knowledge $\kappa$, memory $\mu$, and time $T$ yields an \textsc{Ebhs} algorithm with node storage $\sigma$, knowledge $\kappa$, memory $O(\max\{\mu,\log n\})$ per agent, and time $O(T)$ after the emergence of the black hole.

As an example, a single agent can explore a static graph using a DFS traversal with $O(\log \Delta)$ node storage, $O(\log n)$ agent memory, and $O(m)$ rounds. By replacing each move of the DFS traversal with the cautious chain movement, we obtain a solution to \textsc{Ebhs} that uses the same $O(\log \Delta)$ node storage, $O(\log n)$ agent memory, and $O(m)$ rounds.
\end{remark}

\section{Conclusion}
\textsc{Algo\_Bhs\_Scattered} solves the 1-BHS problem on 1-bounded 1-interval connected TVGs starting with $2\delta_{BH}+17$ many scattered agents. One can also look for solutions on weakly-connected dynamic graphs such as temporal graphs. We provide a solution for the \textsc{Ebhs} problem on arbitrary static graphs using 4 co-located agents without using any other model assumptions extending  the existing literature on rings. Three agents can solve \textsc{Ebhs} on arbitrary graphs or not remains open.  Another direction is to study \textsc{Ebhs} on arbitrary dynamic graphs, starting from strongly connected dynamic graphs (1-interval connected) to weakly connected dynamic graphs (temporal connectivity).

\section*{Acknowledgement}
    Tanvir Kaur acknowledges the support from the CSIR, India (Grant No. 09/1005(0048)/2020-EMR-I). Ashish Saxena acknowledges the financial support from IIT Ropar. Kaushik Mondal acknowledges the ISIRD grant provided by IIT Ropar.

\bibliography{bib}

@inproceedings{GoswamiBD024,
  author       = {Pritam Goswami and
                  Adri Bhattacharya and
                  Raja Das and
                  Partha Sarathi Mandal},
  editor       = {},
  title        = {Perpetual Exploration of a Ring in Presence of Byzantine Black Hole},
  booktitle    = {OPODIS},
  series       = {},
  volume       = {324},
  pages        = {17:1--17:17},
  publisher    = {},
  year         = {2024},
  url          = {},
  doi          = {},
  bibsource    = {}
}

@inproceedings{Adri_ByzBH,
  author       = {Adri Bhattacharya and
                  Pritam Goswami and
                  Evangelos Bampas and
                  Partha Sarathi Mandal},
  editor       = {},
  title        = {Perpetual Exploration in Anonymous Synchronous Networks with a Byzantine
                  Black Hole},
  booktitle    = {DISC},
  series       = {},
  volume       = {356},
  pages        = {16:1--16:17},
  publisher    = {},
  year         = {2025},
  url          = {},
  doi          = {},
  bibsource    = {}
}

@inproceedings{Bonnet25,
  author       = {Fran{\c{c}}ois Bonnet and
                  Quentin Bramas and
                  Anissa Lamani},
  editor       = {},
  title        = {Brief Announcement: Searching for an Eventually-Emerging Black Hole
                  in Rings},
  booktitle    = {SSS},
  series       = {},
  volume       = {16350},
  pages        = {82--87},
  publisher    = {},
  year         = {2025},
  url          = {},
  doi          = {},
  bibsource    = {}
}

@inproceedings{Paola_2002,
author = {Dobrev, S. and Flocchini, P. and Prencipe, G. and Santoro, N.},
title = {Searching for a black hole in arbitrary networks: optimal mobile agent protocols},
year = {2002},
pages = {153–162},
series = {PODC}
}

@InProceedings{Paola_2004,
author="Dobrev, S.
and Flocchini, P.
and Santoro, N.",
title="Improved Bounds for Optimal Black Hole Search with a Network Map",
booktitle="SIROCCO",
year="2004",
pages="111--122"
}

@article{Pelc_2005,
  author       = {Jurek Czyzowicz and
                  Dariusz R. Kowalski and
                  Euripides Markou and
                  Andrzej Pelc},
  title        = {Searching for a Black Hole in Synchronous Tree Networks},
  journal      = {Comb. Probab. Comput.},
  volume       = {16},
  number       = {4},
  pages        = {595--619},
  year         = {2007}
}

@article{Paola_2010,
  author       = {Balasingham Balamohan and
                  Paola Flocchini and
                  Ali Miri and
                  Nicola Santoro},
  title        = {Time Optimal Algorithms for Black Hole Search in Rings},
  journal      = {Discret. Math. Algorithms Appl.},
  volume       = {3},
  number       = {4},
  pages        = {457--472},
  year         = {2011}
}

@InProceedings{Shantanu_2011,
author="Chalopin, J.
and Das, S.
and Labourel, A.
and Markou, E.",
title="Black Hole Search with Finite Automata Scattered in a Synchronous Torus",
booktitle="DISC",
year="2011",
pages="432--446"
}

@article{Paola_2006,
  author       = {S. Dobrev and
                  P. Flocchini and
                  G. Prencipe and
                  N. Santoro},
  title        = {Searching for a black hole in arbitrary networks: optimal mobile agents
                  protocols},
  journal      = {Distributed Comput.},
  volume       = {19},
  number       = {1},
  pages        = {1--35},
  year         = {2006}
  }

@article{bhattacharya_2023, title={Searching for a Black Hole in a Dynamic Cactus}, volume={29}, number={2}, journal={Journal of Graph Algorithms and Applications}, author={Bhattacharya, Adri and Italiano, Giuseppe F. and Mandal, Partha Sarathi}, year={2025}, month={May}, pages={127–166} }

@inproceedings{Adri_tori,
  author       = {Adri Bhattacharya and
                  Giuseppe F. Italiano and
                  Partha Sarathi Mandal},
  title        = {Black Hole Search in Dynamic Tori},
  booktitle    = {{SAND}},
  volume       = {292},
  pages        = {6:1--6:16},
  year         = {2024}
}

@inproceedings{MarkouP12,
  author       = {E. Markou and
                  M. Paquette},
  title        = {Black Hole Search and Exploration in Unoriented Tori with Synchronous Scattered Finite Automata},
  booktitle    = {OPODIS},
  pages        = {239--253},
  year         = {2012}
}

@article{complexity_black_hole,
  author       = {J. Czyzowicz and
                  D. R. Kowalski and
                  E. Markou and
                  A. Pelc},
  title        = {Complexity of Searching for a Black Hole},
  journal      = {Fundam. Informaticae},
  volume       = {71},
  number       = {2-3},
  pages        = {229--242},
  year         = {2006}, 
}

@article{bhs_tokens,
  author       = {S. Dobrev and
                  N. Santoro and
                  W. Shi},
  title        = {Using Scattered Mobile Agents to Locate a Black Hole in an un-Oriented Ring with Tokens},
  journal      = {Int. J. Found. Comput. Sci.},
  volume       = {19},
  number       = {6},
  pages        = {1355--1372},
  year         = {2008},
  
}

@inproceedings{BHS_gen,
  author       = {T. Kaur and
                  A. Saxena and
                  P. S. Mandal and
                  K. Mondal},
  title        = {Black Hole Search in Dynamic Graphs},
  booktitle    = {{ICDCN}},
  year         = {2025}
}

@article{dyn_ring_jrnl,
  author       = {G. A. D. Luna and
                  P. Flocchini and
                  G. Prencipe and
                  N. Santoro},
  title        = {Locating a black hole in a dynamic ring},
  journal      = {J. Parallel Distributed Comput.},
  volume       = {196},
  pages        = {104998},
  year         = {2025}
}

@article{GOTOH20211,
title = {Exploration of dynamic networks: Tight bounds on the number of agents},
journal = {Journal of Computer and System Sciences},
volume = {122},
pages = {1-18},
year = {2021},
issn = {0022-0000},
author = {Tsuyoshi Gotoh and Paola Flocchini and Toshimitsu Masuzawa and Nicola Santoro}
}

@inproceedings{Kuhn_2010,
author = {Kuhn, F. and Lynch, N. and Oshman, R.},
title = {Distributed computation in dynamic networks},
year = {2010},
booktitle = {Proceedings of the Forty-Second ACM Symposium on Theory of Computing},
pages = {513–522},
numpages = {10}
}

@inproceedings{BHS_SSS25,
  author       = {Tanvir Kaur and
                  Ashish Saxena and
                  Partha Sarathi Mandal and
                  Kaushik Mondal},
  title        = {Black Hole Search by Scattered Agents on Time-Varying Dynamic Graphs},
  booktitle    = {{SSS}},
  volume       = {16350},
  pages        = {309--324},
  year         = {2025}
}

@misc{bonnet_2026,
  author       = {François Bonnet and Quentin Bramas and Anissa Lamani},
  title        = {Searching for an Eventually-Emerging Black Hole in Rings},
  year         = {2026},
  note         = {Unpublished manuscript},
  howpublished = {\url{https://hal.science/hal-05517896}}
}

@article{Dobrev2007BHSRing,
  author    = {Stefan Dobrev and Paola Flocchini and Giuseppe Prencipe and Nicola Santoro},
  title     = {Mobile Search for a Black Hole in an Anonymous Ring},
  journal   = {Algorithmica},
  volume    = {48},
  number    = {1},
  pages     = {67--90},
  year      = {2007},
  doi       = {10.1007/s00453-006-1232-z}
}

@article{Reingold2008,
  author    = {Omer Reingold},
  title     = {Undirected connectivity in log-space},
  journal   = {J. {ACM}},
  volume    = {55},
  number    = {4},
  pages     = {17:1--17:24},
  year      = {2008},
}

\end{document}